%% file: main.tex
\documentclass{article}

\usepackage{PRIMEarxiv}

\usepackage{times}
\usepackage{helvet}
\usepackage{courier}
\usepackage{url}
\usepackage{caption}
\usepackage{graphicx} 
\usepackage{hyperref}
\usepackage{comment}
\usepackage{natbib}

\usepackage[ruled,vlined,linesnumbered]{algorithm2e}
\usepackage{comment}

\usepackage{environ}
\newif\ifshort
\usepackage{newfloat}
\usepackage{listings}
\usepackage{amsthm}
\usepackage{amssymb}
\usepackage{mathtools}
\usepackage{diagbox}
\usepackage{multirow}
\usepackage{graphicx,color}
\usepackage{algorithmic}
\usepackage[most]{tcolorbox}

\usepackage{enumitem}

\usepackage{xspace}
\usepackage{cleveref}
\usepackage{todonotes}
\usepackage{thm-restate}

\newtheorem{theorem}{Theorem}
\newtheorem{lemma}{Lemma}
\newtheorem{claim}{Claim}

\theoremstyle{definition}

\theoremstyle{definition}

\newcommand{\cqed}{\renewcommand{\qedsymbol}{$\lrcorner$}\qed}

\DeclareCaptionStyle{ruled}{labelfont=normalfont,labelsep=colon,strut=off} % DO NOT CHANGE THIS
\input{defs}

\title{Learning Augmented Exact Exponential Algorithms}
\author{
  Tatiana Belova\\
  ITMO University\\
  \And
  Yuriy Dementiev\\
  ITMO University\\
  \And
  Danil Sagunov\\
  ITMO University\\
}
\date{}

\begin{document}

\maketitle

\begin{abstract}
    The field of learning-augmented algorithms has demonstrated that machine-learned predictions can bypass worst-case lower bounds across a wide range of problems. So far, however, the focus has been almost exclusively on \emph{polynomial-time} algorithms, where predictions improve competitive ratios, approximation guarantees, or running times.  In this paper, we raise the question of whether predictions can push the frontier of \emph{exact exponential-time} algorithms for \classNP-hard problems.
    We answer this question affirmatively by proposing a \emph{general approach} that augments an entire family of state-of-the-art exact algorithms for a variety of subset selection problems.
    We show that a noisy predictor that is only marginally better than random guessing suffices to provably reduce the search space, and that the resulting runtime speedup scales smoothly with the prediction quality.
    Importantly, our algorithms require only \emph{pairwise} independence of predictions or, alternatively, do not require the knowledge of the predictor's accuracy---both strictly weaker and more realistic settings than typically assumed.
\end{abstract}

\section{Introduction}\label{sec:intro}
Many fundamental combinatorial optimisation problems---including \textsc{Feedback Vertex Set}, \textsc{Hitting Set}, and \textsc{Weighted $d$-SAT}---are \classNP-hard, and are therefore unlikely to admit polynomial-time exact algorithms. The best known exact algorithms for such problems run in time $\alpha^n \cdot n^{\Oh(1)}$ for some base $\alpha > 1$ that depends on the problem. 
Improving $\alpha$ even slightly is a major algorithmic achievement: shaving the base from $2$ to $1.9$ or from $1.7$ to $1.6$ represents an exponential gain in practice and requires fundamentally new combinatorial ideas. 
At the same time, the \emph{Strong Exponential Time Hypothesis} (SETH)~\citep{SETH} provides strong evidence that such improvements have limits: it asserts that for every $\varepsilon > 0$, there exists $k$ such that $k$-\textsc{SAT} cannot be solved in time
$\Oh((2-\varepsilon)^n)$, and reductions propagate this barrier to a broad family of subset problems.

This raises a natural question that has not been studied before: can \emph{machine-learned predictions} about the structure of the optimal solution push the frontier of exact exponential-time algorithms---and perhaps even circumvent the SETH barrier entirely?

\paragraph{Learning-augmented algorithms.}
A growing line of work studies how machine-learned predictions can be integrated into algorithm design while maintaining rigorous worst-case guarantees. 
This paradigm, known as \emph{learning-augmented algorithms} (LAA), was pioneered by~\citet{LV18} and~\citet{Pur18}, who showed that predictions about the unknown future allow online algorithms to bypass information-theoretic lower bounds. Since then, the paradigm has been applied with great success to caching~\citep{LV18}, scheduling~\citep{Pur18}, ski rental, and dynamic power management~\citep{Ant21}, with recent work refining the consistency--robustness tradeoff~\citep{El24}, integrating the learning of predictors directly into the algorithmic task~\citep{Elias24}, and extending the framework to mechanism design~\citep{Chr24} and Markov decision processes~\citep{Li23}.

More recently, predictions have been brought to bear on offline combinatorial optimization, where the goal is to circumvent hardness-of-approximation barriers rather than information-theoretic ones. 
\citet{GPS22} showed that $\varepsilon$-accurate predictions can break competitive ratio lower bounds, and the landmark work of \citet{CDGLP24} established that noisy element-wise predictions suffice to beat classical approximation hardness results for \textsc{MaxCut} and constraint satisfaction problems.
Subsequent work extended these ideas to \textsc{Vertex Cover}, \textsc{Set Cover}, and \textsc{Maximum Independent Set}~\citep{Aam25}, as well as to dense instances of \classNP-hard problems~\citep{BEX24}.
% Predictions have also been used to accelerate polynomial-time algorithms more broadly, improving running times for frequency estimation~\citep{Hsu19}, data stream algorithms~\citep{Jiang20}, and fundamental graph algorithms, as well as enabling algorithms to learn their own predictors from data~\citep{Kho22}.

Despite this remarkable progress, all existing LAA results concern algorithms that run in \emph{polynomial time}. 
Predictions have been used to improve constant factors of running time, competitive ratios, or approximation guarantees---but never to accelerate an algorithm whose baseline complexity is exponential. 
The question of whether predictions can push the frontier of exact exponential-time algorithms for \classNP-hard problems has remained entirely open. This is the question we address in this paper.

%\vspace*{-2em}

\paragraph{The information-theoretic perspective.}
It is instructive to compare our setting with prior LAA work on \classNP-hard problems through an information-theoretic lens. In existing work, the predictor provides from $n$ to $n^2$ bits of information---one binary prediction per element of the input (or per pair of elements)---and these bits are used to improve the output quality of a polynomial-time algorithm. In our setting, the predictor provides the same $n$ bits, yet the algorithmic task is fundamentally harder in two respects.
First, we seek an \emph{exact} solution: unlike approximation algorithms, our algorithms have no margin for error and must return a provably optimal subset. Second, the baseline algorithm operates in
\emph{exponential time}, systematically exploring an exponential search space and generating an exponential amount of internal information in the process. It is therefore far from obvious that $n$ bits of external prediction can have any meaningful impact on such a search.

Our results show that this intuition is wrong. Even a predictor that is barely better than random guessing---correct with probability only $\frac{1}{2} + \varepsilon$ per element---suffices to provably reduce the effective search space by an exponential factor. This reveals an unexpected \emph{information leverage effect}: a linear amount of predictive information is sufficient to tame an exponential search, provided it is exploited through the right algorithmic framework.

\paragraph{Setting.}
We study \classNP-hard problems in which the goal is to find a subset $S$ of a given universe $\mathcal{U} = [n]$ satisfying some combinatorial feasibility condition. We formalise this through an implicit set system $\Phi$, which maps each problem instance $I$ to a pair $(U_I, \mathcal{F}_I)$, where $U_I=[n]$ is the universe
and $\mathcal{F}_I \subseteq 2^{U_I}$ is the family of feasible solutions. 
The $\Phi$-\textsc{Subset} problem asks to find a set $S \in \mathcal{F}_I$, and the
$\Phi$-\textsc{Extension} problem asks whether a given partial set $X \subseteq U_I$ can be extended to a feasible solution by adding at most $k$ further elements.

\defparproblema{\textsc{$\Phi$-Subset}}{An instance $I$.}{}{Find a set $S \in \mathcal{F}_I$ if one exists.}

\defparproblema{\textsc{$\Phi$-Extension}}{An instance $I$, a set $X \subseteq [n]$, and an integer $k$.}{}{Does there exists a subset $Y \subseteq ([n] \setminus X)$ such that $X \cup Y \in \mathcal{F}_I$ and $|Y| \leq k$?}

\begin{example}
\textsc{Feedback Vertex Set} is a $\Phi$-\textsc{Subset}
problem where $U_I$ is the vertex set of the input graph $G$,
and $\mathcal{F}_I$ is the family of all subsets $S \subseteq
V(G)$ of size at most $k$ such that $G - S$ is acyclic.
The corresponding $\Phi$-\textsc{Extension} problem asks:
given a partial set $X \subseteq V(G)$, does there exist a
set $Y \subseteq V(G) \setminus X$ of size at most $k$ such
that $G - (X \cup Y)$ is acyclic?
\end{example}
To incorporate predictions into our algorithms, we adopt the
\emph{noisy predictor} framework. Let $I$ be an instance of a subset selection problem with universe $[n]$, and let $S^\star \subseteq [n]$ be an optimal solution. 
A \emph{noisy predictor} provides for each element $e \in U$ a binary prediction $\hat{y}_e \in \{0,1\}$ of whether $e \in S^\star$. We require that each prediction is correct with probability at least $\frac{1}{2} + \varepsilon$ for some $\varepsilon > 0$, and that
prediction errors are \emph{pairwise independent} across elements.

Several design choices in this model deserve discussion, as they distinguish our setting from prior work. 
The first concerns the \emph{independence structure} of the predictions. The standard
assumption in the literature---adopted, for instance, by \citet{NoisySorting}---is that prediction errors are \emph{mutually independent}. 
We instead require only \emph{pairwise independence}, which is strictly weaker.
This distinction matters in practice, as pairwise independence can be guaranteed by simple constructions such as hash-based predictors, whereas mutual independence is a much stronger structural requirement that is difficult to verify or enforce.
The second concerns \emph{knowledge of the prediction quality} $\varepsilon$. Most theoretical analyses assume that $\varepsilon$ is known to the algorithm, which allows it to be tuned optimally. 
In practice, however, the accuracy of a learned predictor is rarely known precisely in advance. We therefore also consider a regime in which $\varepsilon$ is unknown to the algorithm, and show that our approach remains effective in this setting.

\paragraph{Monotone local search.}
To exploit element-wise predictions algorithmically, we need a framework that explores the space of feasible subsets in a structured, prediction-amenable way. 
Concretely, we need a framework that builds solutions \emph{incrementally}---adding elements one by one---so that predictions about which elements belong to the optimal solution can be used to bias
the search towards promising regions of the solution space. 
The \emph{monotone local search} (MLS) framework of~\citet{MLS} provides exactly this structure.

% The MLS framework is built around the \textsc{$\Phi$-Extension} problem: given a partial solution $X \subseteq U$ and an integer $k$, does there exist a set $S \subseteq U \setminus X$ with $|S| \leq k$ such that $S \cup X$ is a feasible solution? 

The key insight of this framework is that any fixed-parameter tractable (\classFPT) algorithm for \textsc{$\Phi$-Extension} can be converted into an exact algorithm for the original \textsc{$\Phi$-Subset} problem.
Exact exponential-time algorithms are actively used in practice, 
with applications ranging from formal verification and bioinformatics  to competitive benchmarks such as the PACE Challenge~\citep{PACE, FominKratsch}.

\begin{theorem}[\citealt{MLS}]
If \textsc{$\Phi$-Extension} can be solved in time $c^k \cdot n^{\Oh(1)}$, then \textsc{$\Phi$-Subset} can be solved by a randomised algorithm in time $\bigl(2 - \tfrac{1}{c}\bigr)^n \cdot n^{\Oh(1)}$.
\end{theorem}

This single theorem, applied to state-of-the-art \classFPT algorithms,immediately yields the best known exact algorithms for a broad family of \classNP-hard subset selection problems. 

% \begin{table}[ht]
% \centering
% \caption{Summary of known result for different optimization problem.}
% \label{tab:problems}
% \renewcommand{\arraystretch}{1.25}
% \begin{tabular}{@{}lll@{}}
% \toprule
% \textbf{Problem} & \textbf{FPT bound} & \textbf{SOTA exact} \\
% \midrule
% \textsc{Feedback Vertex Set} & $3^k$ (r)~\citep{Cygan15}
%   & $1.6667^n$ (r) \\
% \textsc{Feedback Vertex Set} (det.) & $3.592^k$~\citep{KP14}
%   & $1.7217^n$ \\
% \textsc{Subset Feedback Vertex Set} & $4^k$~\citep{Cygan13}
%   & $1.7500^n$ \\
% \textsc{Group Feedback Vertex Set} & $4^k$~\citep{Cygan13}
%   & $1.7500^n$ \\
% \textsc{FVS in Tournaments} & $1.6181^k$~\citep{KL16}
%   & $1.3820^n$ \\
% $d$-\textsc{Hitting Set} & $d^k$
%   & $(2-\tfrac{1}{d})^n$ (r) \\
% $3$-\textsc{Hitting Set} & $2.076^k$~\citep{Wahlstrom07}
%   & $1.6667^n$ (r) \\
% \textsc{Node Unique Label Cover} & $|\Sigma|^{2k}$~\citep{FLRSV12}
%   & $(2-\tfrac{1}{|\Sigma|^2})^n$ (r) \\
% \textsc{Min-Ones} $d$-\textsc{SAT} & $d^k$
%   & $(2-\tfrac{1}{d})^n$ (r) \\
% \textsc{Min-Ones} $3$-\textsc{SAT} & $2.562^k$~\citep{Kneis06}
%   & $1.6097^n$ (r) \\
% \textsc{Interval Vertex Deletion} & $8^k$~\citep{Cao16}
%   & $1.8750^n$ \\
% \textsc{Proper Interval Vertex Deletion} & $6^k$~\citep{Cao15}
%   & $1.8334^n$ \\
% \textsc{Block Graph Vertex Deletion} & $4^k$~\citep{Agrawal16}
%   & $1.7500^n$ \\
% \textsc{Cluster Vertex Deletion} & $1.9102^k$~\citep{Bocker12}
%   & $1.4765^n$ (r) \\
% \textsc{Vertex $(r,\ell)$-Partization} ($r,\ell\leq 2$)
%   & $3.3146^k$~\citep{Baste15}
%   & $1.6984^n$ (NPR) \\
% \bottomrule
% \end{tabular}
% \end{table}

\paragraph{Our results and contributions.}
We develop a learning-augmented approach for exact exponential-time
algorithms based on the noisy predictor model, and prove two main
results.

Our first result concerns the general case where no \classFPT subroutine is available. We give a prediction-guided algorithm that solves any $\Phi$-\textsc{Subset} problem in time $(2 - \Omega(\varepsilon^2))^n
\cdot n^{\Oh(1)}$, assuming only an oracle for feasibility checking.
This strictly improves upon the trivial $2^n$ brute-force bound for
any $\varepsilon > 0$, and in particular breaks the running-time
SETH barrier.

Our second result concerns the MLS setting. We integrate the noisy
predictor into the MLS framework and show that for any problem
admitting a $c^k$-time \classFPT algorithm for $\Phi$-\textsc{Extension}, the bound of $(2 - \frac{1}{c})^n$ of \citet{MLS} is strictly improved for any $\varepsilon > 0$  to $(2 - \frac{1}{c}-\Omega_c(\varepsilon^2))^n$, where $\Omega_c$ hides a constant greater than $0$ dependent on $c$. This approach applies to any problem that admits a fixed-parameter tractable extension algorithm, and thus simultaneously improves the state-of-the-art exact algorithms for \textsc{Feedback Vertex Set} (randomized %~\citep{Cygan15} 
and
deterministic) %~\citep{KP14}), 
\textsc{Subset Feedback Vertex
Set}, %~\citep{Cygan13},
\textsc{Group Feedback Vertex
Set}, %~\citep{Cygan13}, 
\textsc{Feedback Vertex Set in
Tournaments}, %~\citep{KL16}, 
$d$-\textsc{Hitting Set}, 
$3$-\textsc{Hitting Set}, %~\citep{Wahlstrom07}, 
\textsc{Node Unique
Label Cover}, %~\citep{FLRSV12}, 
\textsc{Min-Ones} $d$-\textsc{SAT}, 
\textsc{Min-Ones} $3$-\textsc{SAT}, %~\citep{Kneis06}, 
\textsc{Interval
Vertex Deletion}, %~\citep{Cao16},
\textsc{Proper Interval Vertex
Deletion}, %~\citep{Cao15}, 
\textsc{Block Graph Vertex
Deletion}, %~\citep{Agrawal16}, 
\textsc{Cluster Vertex
Deletion}, %~\citep{Bocker12},
and \textsc{Vertex}
$(r,\ell)$-\textsc{Partization} for $r,\ell \leq 2$.%~\citep{Baste15}.

In each case, the improvement is by reducing the  factor by at least
$\Omega_c(\varepsilon^2)$ in the base of the exponent.

Both results hold under pairwise independence of prediction errors.
Finally, when the prediction quality $\varepsilon$ is unknown, we show that under mutual independence the same asymptotic running times are achievable in expectation, by running the algorithm with geometrically decreasing guesses of $\varepsilon$.

\paragraph{Related work.}
The learning-augmented algorithms framework was introduced by \citet{LV18} for the caching problem and by \citet{Pur18} for ski rental and scheduling.
These works established the central desiderata of \emph{consistency} (near-optimal performance when predictions are correct) and \emph{robustness} (worst-case guarantees when predictions are adversarial), and the consistency--robustness tradeoff has since become the standard evaluation criterion for LAA. 
The paradigm has been extended to dynamic power management~\citep{Ant21}, where authors derive tight ski rental bounds as a core ingredient; to the Bahncard problem~\citep{Zhao24}; and to mechanism design~\citep{Chr24}. 
\citet{Elias24} depart from the black-box prediction paradigm by integrating the learning of the predictor directly into the algorithmic task, obtaining improved guarantees for scheduling and caching. 
\citet{El24} demonstrate that Pareto-optimal consistency---robustness tradeoffs can be extremely brittle under small prediction errors, and propose a smoothness-enforced framework to remedy this. 
\citet{Li23} study learning-augmented algorithms for time-varying MDPs where the advice consists of $Q$-values rather than black-box recommendations, showing that structural information about the predictor enables stronger guarantees. A comprehensive survey of the area is given by \citet{MV22}.

\emph{Predictions for \classNP-hard offline problems.}
A more recent and directly relevant thread applies predictions to offline combinatorial optimisation, aiming to circumvent hardness-of-approximation barriers. 
\citet{GPS22} introduce $\varepsilon$-accurate predictions and show that even mildly accurate predictions can break competitive ratio lower bounds for online algorithms. 
The landmark work of \citet{CDGLP24} establishes an $\varepsilon$-prediction framework for \classNP-hard offline problems: element-wise membership predictions correct with probability $\frac{1}{2} + \varepsilon$ suffice to beat classical approximation hardness for \textsc{MaxCut} and constraint satisfaction problems. 
Subsequent work extends these ideas to \textsc{Vertex Cover}, \textsc{Set Cover}, and \textsc{Maximum Independent Set} via edge-based predictions~\citep{Aam25}, to combinatorial optimisation with predictions more broadly~\citep{AEP24}, and to dense instances of \classNP-hard problems~\citep{BEX24}. 
\citet{GhMM25} study constraint satisfaction problems with advice.

\paragraph{Organization.}
Section~\ref{sec:search} presents both algorithms and their analysis,
including the proofs of the two main theorems. Section~\ref{sec:quality}
addresses the setting where the prediction quality $\varepsilon$ is
unknown.

\section{Prediction-Guided Search Space}\label{sec:search}

In this section, we show how predictions help reducing the combinatorial search space and speeding up exhaustive search as well as more sophisticated methods.

\subsection{Exhaustive Search with Predictions}\label{subsec:exhaustive}
We start with a general subset selection problem from the class \classNP.
We can consider this as a problem where we are given an integer $n$ and a polynomial-time algorithm $\mathcal{A}$.
It takes a subset of $[n]$ as an input and reports $1$ or $0$.
Our goal is to find any subset $S\subset [n]$ with $\mathcal{A}(S)=1$.
Without any additional insight, this problem is trivially resolvable in $2^n\cdot \polyn$ running time by enumerating all $S\subset [n]$ and evaluating $\mathcal{A}(S)$ in polynomial time.

In the noisy prediction model, the predictor makes predictions about some fixed solution $S^*\subset [n]$ with $\mathcal{A}(S^*)=1$.
For each element of $i\in [n]$, when we ask the predictor \emph{``Does $S^*$ contain $i$?''} it gives the correct answer with probability at least $\frac{1}{2}+\varepsilon$.
By asking this question once for each separate element $i\in[n]$, we construct a set $\widetilde{S}^*$ consisting of all ``yes''-elements.
On average, at least $\frac{1}{2}+\varepsilon$ fraction of guesses were correct.
Therefore, the average symmetric difference between $S^*$ and the prediction set is $\mathbb{E}\left[|S^*\triangle\widetilde{S}^*|\right]\le (\frac{1}{2}-\varepsilon)\cdot n$, and is separated from $\frac{n}{2}$ with $\varepsilon n$ gap.
This suggests that we could use $\widetilde{S}^*$ as a ``hot start'' set and enumerate all symmetric difference sets $D\subset [n]$ of size at most $(\frac{1}{2}-\Omega(\varepsilon))\cdot n$.
If $\mathcal{A}(\widetilde{S}^*\triangle D)$ equals $1$, the solution is found.

The intuition above raises two concerns.
First of them is why does this strategy (enumerating subsets of slightly limited size) even provide any substantial improvement over $2^n\cdot\polyn$.
The answer to this question is a well-known upper bound on binomial coefficients $\binom{n}{\lfloor \gamma n \rfloor}\le 2^{H(\gamma)\cdot n}$, where \[H(\gamma)=-\gamma\cdot \log_2\gamma-(1-\gamma)\cdot \log_2(1-\gamma)\] is the Shannon's binary enthropy function, and the inequality holds for any $\gamma\in(0,\frac{1}{2}]$.
When $\gamma$ is separated from $\frac{1}{2}$, $2^{H(\gamma)}$ is separated from $2$.
We formalize this in the following lemma.

\begin{lemma}\label{lemma:enthropy-near-half}
    For every $x \in [0, \frac{1}{2}]$, $2^{H(\frac{1}{2} - x)} \le 2 - \frac{11}{3}x^2$.
\end{lemma}
\begin{proof}
We want to estimate how far $2^{H(\frac{1}{2} - x)}$ stands from $2$ with a small-degree polynomial in $x$ on the interval $[0, \frac{1}{2}]$.
As we can see in \Cref{fig:graph_h(1/2-x)}, this gap is close to $4x^2$ on that interval. 
In the following, we estimate it from below by $\frac{11}{3} \cdot x^2$.

\begin{figure}[ht]
    \centering
    \includegraphics[width=0.9\linewidth]{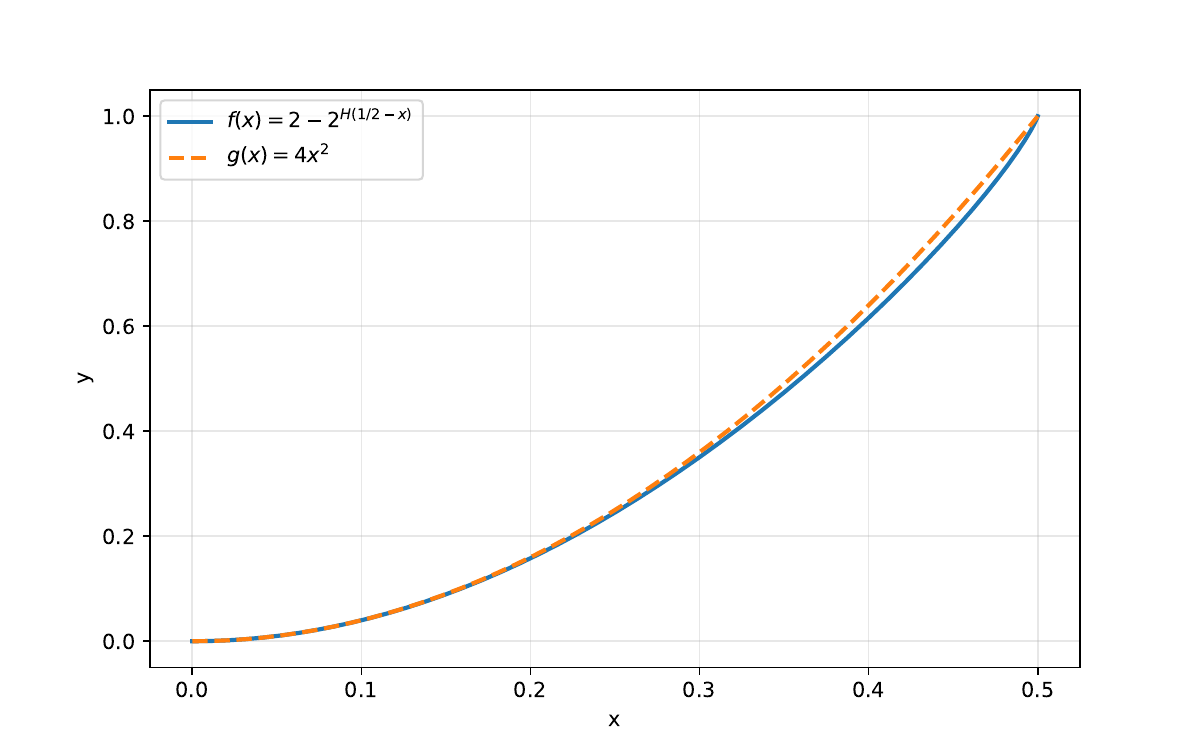}
    \caption{$2 - 2^{H(\frac{1}{2}-x)}$ and $4x^2$ on $[0, \frac{1}{2}]$}
    \label{fig:graph_h(1/2-x)}
\end{figure}

    The Taylor series of the $H(p)$ at $\frac{1}{2}$ is
    $$H(p) = 1 - \frac{1}{2 \ln 2}\sum\limits_{n=1}^{\infty} \frac{(1 - 2p)^{2n}}{n (2n - 1)} \; ,$$
    which converges to $H(p)$ for every $p \in [0, 1]$.
    By setting $p := \frac{1}{2} - x$, we get
    $$H\left(\frac{1}{2} - x\right) = 1 - \frac{1}{2 \ln 2}\sum\limits_{n=1}^{\infty} \frac{(2x)^{2n}}{n (2n - 1)}\; .$$    
    Hence,
    \[
    2^{H(\frac{1}{2} - x)} =
    2^{1 -\frac{1}{2 \ln 2}\sum\limits_{n=1}^{\infty} \frac{(2x)^{2n}}{n (2n - 1)}} =
    2 \cdot e^{-\frac{1}{2}\sum\limits_{n=1}^{\infty} \frac{(2x)^{2n}}{n (2n - 1)}} = 
    2 \cdot e^{- 2x^2 - \frac{4}{3} x^4 - \frac{32}{15} x^6 - \dots} .
    \]
    The Taylor series for $e^t$ at $0$ is 
    $$e^t = \sum\limits_{n=0}^{\infty} \frac{t^i}{n!} ,$$
    which converges to $e^t$ for every $t$.
    We now substitute $t = - 2x^2 - \frac{4}{3} x^4 - \frac{32}{15} x^6 - \dots$ and get
    $$2^{H(\frac{1}{2}-x)} = 2 - 4x^2 + \frac{4}{3} x^4 - \frac{8}{5} x^6 + \dots \; .$$
    This series converges for $x\in[0,1/2]$ because $2^{H(1/2-x)}$ is analytic.  
    Moreover, for $0 \le x \le 1/2$ the coefficients alternate in sign and their absolute values decrease after the first few terms.  
    Consequently, for any $x$ in this interval, the partial sum up to the $x^4$ term provides an upper bound:
    \[
    2^{H(\frac{1}{2} - x)} \le 2 - 4x^2 + \frac{4}{3}x^4.
    \]
    Since $x^2 \le 1/4$ on $[0,1/2]$, we have $x^4 \le x^2/4$, and therefore
    \[
    2^{H(\frac{1}{2} - x)} \le 2 - 4x^2 + \frac{4}{3}\cdot\frac{x^2}{4} = 2 - \frac{11}{3}\,x^2.
    \]
    This completes the proof.
\end{proof}

Consequently, enumerating symmetric difference sets accordingly to above takes at most $(2-\Omega(\varepsilon^2))^n\cdot\polyn$ time.
Second concern is related to the probability of that the prediction is accurate enough that this strategy can be applied. 
If the guesses are pairwise independent, the probability that the size of $D$ is too far from the average is inversely proportional to $n$.
Under mutual independence, this probability drops exponentially in $n$.
%\todo[inline]{probably note here that we cannot avoid independence since the guesses are garbage?}
\begin{lemma}\label{lemma:sum_near_expectation}
Let $\varepsilon \in (0, \frac{1}{2})$,  
and let $P_1, \dots, P_n \in \{0,1\}$ be binary random variables such that for every $i \in [n]$,
$\Pr \left[P_i = 0 \right] \le \frac{1}{2} - \varepsilon$.
Then, for $X = \sum_{i=1}^n P_i$,
\[
\Pr\left[X \le \left(\tfrac{1}{2} + \tfrac{\varepsilon}{2}\right) \cdot n\right] \le
\begin{cases}
 \tfrac{1}{\varepsilon^2 n}, & \text{if } P_1,\dots,P_n \text{ are \emph{pairwise} independent}, \\[1.2em]
e^{-\varepsilon^2 n / 2},   & \text{if } P_1,\dots,P_n \text{ are \emph{mutually} independent}.
\end{cases}
\]
\end{lemma}
\begin{proof}
Let $X_i$ be a random variable that signals if the oracle $P_i$ correctly guesses if $i \in S$, so 
$$X_i = \big[P_i = [i \in S] \, \big]  \text{ and } X = \sum\limits_{i=1}^n X_i \; .$$
We know that $\mathbb{E} [X_i] = \Pr [X_i] \ge \frac{1}{2} + \varepsilon$, so $\mathbb{E}[X] \ge (\frac{1}{2} + \varepsilon) \cdot n$.
Hence, 
\[
\Pr\left[X < \Bigl(\frac{1}{2} + \frac{\varepsilon}{2}\Bigr) \cdot n\right] \le 
\Pr\left[X < \mathbb{E} [X] - \frac{\varepsilon n}{2} \right] \le
\Pr\left[\big |X - \mathbb{E} [X] \big| > \frac{\varepsilon n}{2} \right] \;.
\]

We now estimate this probability using the Chebyshev inequality for the case of pairwise independence, and the Chernoff-Hoeffding's inequality for the case of full independence.

By the Chebyshev inequality, 
$$\Pr[|X - \mathbb{E}[X]| > \frac{\varepsilon n}{2}] \le \frac{\Var(X)}{\varepsilon^2 n^2/4} \;.$$

We expand $\Var(X)$ as $\mathbb{E}[X^2] - (\mathbb{E} [X])^2$, and consider those parts separately.
First, 
$$\mathbb{E}[X^2] = \mathbb{E} \left[\sum\limits_{i=1}^n X_i \right]^2 = \sum\limits_{i, j \in [n]} \mathbb{E} [X_i X_j] = 
    \sum\limits_{i=1}^n \mathbb{E} [X_i^2] + \sum\limits_{i \neq j} \mathbb{E} [X_i \cdot X_j] \; .$$
    If $P_i$ are pairwise independent, then $X_i$ are also pairwise independent, so $$\mathbb{E} [X_i X_j] = \mathbb{E} [X_i] \cdot \mathbb{E} [X_j] \text{ \; for $i \neq j$.}$$
    At the same time, since $X_i \in \{0, 1\}$, $\mathbb{E} [X_i^2] = \mathbb{E} [X_i]$.
    That is, 
    $$\mathbb{E} [X^2] = \sum\limits_{i=1}^n \mathbb{E} [X_i] + \sum\limits_{i \neq j} \mathbb{E} [X_i] \cdot \mathbb{E} [X_j] .$$

    Now we expand $\mathbb{E} [X]^2 = \mathbb{E} \left[\sum X_i\right]^2 = 
    \sum\limits_{i = 0}^n \mathbb{E}[X_i]^2 + \sum\limits_{i \neq j} E[X_i] \cdot \mathbb{E}[X_j].$

    Hence,
    $$\Var(X) = \sum\limits_{i=1}^n \!\big(\,\mathbb{E}[X_i] - \mathbb{E}[X_i]^2 \;\!\big) \le \frac{n}{4}\; .$$
    The last inequality holds since the function $x - x^2 \le \frac{1}{4}$ for $x \in [0, 1]$

    Finally, for pairwise independent $P_1, \dots, P_n$, we have
    $$\Pr\left[\sum\limits_{i=1}^n X_i < (\frac{1}{2} + \frac{\varepsilon}{2}) \cdot n\right] \le 
    \Pr\left[\big|\sum X_i - \mathbb{E}\big[\sum\limits_{i=1}^n X_i\big]\big| > \frac{\varepsilon n}{2}\right] \le
    \frac{n/4}{\varepsilon^2 n^2 / 4} = \frac{1}{\varepsilon^2 n} \; ,$$
    so we proved the first part of the lemma.

    We now show the inequality for the case of mutually independence.
    Apply the Chernoff-Hoeffding's inequality for $X_i \in [0,1]$,
    \[
    \Pr \left[\sum X_i \le \mathbb{E} \big[\sum X_i \big] - t \right] \le e^{-2t^2/n} \; ,
    \]
    with $t = \frac{\varepsilon n}{2}$, and get
    \[
    \Pr \big[X < \big(\frac{1}{2} + \frac{\varepsilon}{2}\big) \cdot n \big] \le e^{-\varepsilon^2 n/2} \; ,
    \]
    so we proved the second part of the lemma.
\end{proof}

Combining the strategy with two lemmas above clearly gives us an algorithm that works in $(2-\Omega(\varepsilon)^2)^n\cdot \polyn$ time and finds $S$ with $\mathcal{A}(S)=1$, if the prediction quality is good enough.
The probability that the prediction is not accurate is low when $n>>\frac{1}{\varepsilon^2}$ and tends to zero when $n$ grows.

% \todo[inline]{discussion here that $S$ always exists}
% \todo[inline]{discussion of the susceptibility to bad prediction quality }

\subsection{Monotone Local Search with Predictions}\label{subsec:mls}
Many of the \classNP-hard problems that fit under the general notion of Section~\ref{subsec:exhaustive} are known to admit exact exponential algorithms much more efficient than the exhaustive search.
This is usually attained by the recursive \emph{branching} approach that keeps the search exhaustive but avoids some choices that lead to suboptimal solutions.
Designing a set of branching rules for a certain problem is usually very specific to the problem itself.

\emph{Monotone Local Search} (MLS) is an alternative approach to design of exact exponential algorithms.
To use this framework, one needs to solve a related subset selection problem: given a partial solution $X\subset [n]$, and an integer $k$, find a set $Y\subset[n]\setminus X$ of size at most $k$ such that $X\cup Y$ is a solution, or determines that such $Y$ does not exist.
Designing an algorithm that would work in $c^k\cdot \polyn$ time for constant $c>1$ as small as possible is one of the questions raised by Parameterized Complexity, a very well developed area of algorithm design.
Using MLS, we are able to wrap any such parameterized algorithm up into an exact exponential algorithm for the same problem that works in time $(2-\frac{1}{c})^n\cdot \polyn$.
The work \citep{MLS} that introduced MLS, also demonstrated that this framework improves the best known running times for a list of problems, including those having no faster-than-$2^n$ exact algorithms known before, thanks to the existing advances of Parameterized Complexity.

% \todo[inline]{example of min ones d sat? explain that designing fpt algorithms comes simpler}

In this part of the section, we enhance this powerful framework with predictions.
To understand how it is possible, we have to explain the basic idea of MLS.
Given $\mathcal{B}$, the parameterized algorithm for the extension version of the problem that runs in $c^k\cdot\polyn$ time, we transform it into exact exponential algorithm $\mathcal{E}$.
First of all, $\mathcal{E}$ focuses on finding a solution subset of fixed size $k$.
If $k$ is not known in advance, we can just iterate $k$ in $\{0,1,\ldots, n\}$.
Then, depending on $n$, $k$ and constant $c>1$ that comes from the running time bound of $\mathcal{B}$, the algorithm $\mathcal{E}$ evaluates integer $t$ with $0\le t \le k$.
After that, $\mathcal{E}$ makes a random guess of the partial solution set $X\subset [n]$ of size exactly $t$.
Then, it runs $\mathcal{F}$ on $X$ with parameter $k-t$, aiming to find $Y$ such that $X\cup Y$ is a solution.
Repeating random guesses of $X$ for enough many times guarantees that $\mathcal{E}$ will find a solution with high probability.

If $S^*$ is a solution of size $k$, the probability that a guess of $X$ is correct, that is, $X\subset S^*$, is at least $\frac{\binom{k}{t}}{\binom{n}{t}}$, and $\mathcal{E}$ has to take as many as $\frac{\binom{n}{t}}{\binom{k}{t}}\cdot \polyn$ random samples of $X$ to reach good probability.
Therefore, the running time of $\mathcal{E}$ is upper bounded (up to a polynomial factor) by $\max_{k=0}^n\min_{t=0}^k \frac{\binom{n}{t}}{\binom{k}{t}}\cdot c^{k-t}$, that, in its turn, is upper bounded by $(2-\frac{1}{c})^n$, as shown in \citet{MLS}.
Another great achievement of \citet{MLS} is a construction that allows to avoid randomness in sampling $X$.
In fact, when $n, k$ and $t$ are fixed, we can construct a representative family of subsets deterministically, and the size of this family is almost identical to the number of random samples we have to make to attain good probability.
Then it is enough to take samples of $X$ from this family alone.

\begin{lemma}[Theorem 3.3 in \citep{MLS}]\label{lemma:universal-family}
    There is an algorithm that, given three integers $n\ge k\ge t\ge 0$, constructs a family $\mathcal{S}$ of $t$-element subsets of $[n]$ such that 
    \begin{itemize}
        \item For every $S\subset [n]$ with $|S|=k$, there exists $X\in\mathcal{S}$ such that $X\subset S$;
        \item $\mathcal{S}$ contains at most $\frac{\binom{n}{t}}{\binom{k}{t}}\cdot 2^{o(n)}$ sets;
    \end{itemize}
    $\mathcal{S}$ is called $(n,k,t)$-set-inclusion family and the algorithm works in $|\mathcal{S}|\cdot 2^{o(n)}$ time.
\end{lemma}

The detailed discussion of how MLS works is finished.
The main result of our work is an improvement of MLS in presence of a noisy predictor.

\begin{theorem}\label{thm:noisy-mls}
    If \textsc{$\Phi$-Extension} can be solved in time $c^k\cdot\polyn$ and there exists a noisy predictor with accuracy $\frac{1}{2}+\varepsilon$, then \textsc{$\Phi$-Subset} can be solved in time $(2-\frac{1}{c}-\Omega_c(\varepsilon^2))^n\cdot\polyn$. 
\end{theorem}

In the statement above, $\Omega_c$ hides a constant that depends on $c$.
Naturally, the closer $2-\frac{1}{c}$ is to $1$, the smaller is the enhancing effect of the noisy predictor.
We should also clarify the randomness in \Cref{thm:noisy-mls}.
While our algorithm is deterministic, it relies on the prediction quality.
Using \Cref{lemma:sum_near_expectation} we will show that this quality is good enough with high probability.

We move on to the main idea of our proof of \Cref{thm:noisy-mls}.
The intuition behind it is that the prediction should help us sample the partial set $X$ not totally at random, but with some substantial bias. 
Note that we would like to still keep the algorithm itself deterministic as the original MLS.
To reach this, let us again ask the predictor once about each particular element in $[n]$ of whether it belongs to the solution $S^*$ or not.
The answers divide the universe $\mathcal{U}=[n]$ into two disjoint sets, which we refer to as $\mathcal{U}_0$ and $\mathcal{U}_1$, the ``no''- and ``yes''-answers correspondingly.
The predictor tends to keep the elements of $S^*$ in $\mathcal{U}_1$, and elements of $\mathcal{U}\setminus S^*$ in $\mathcal{U}_0$.
Therefore, we can expect that $\mathcal{U}_1$ contains substantially more than the half of all elements of $S^*$, and substantially less than the half of all elements outside of $S^*$.
Then, when we sample $X$ of size $t$, we sample a part of $X$ inside $\mathcal{U}_0$ and a part of $X$ inside $\mathcal{U}_1$.
The sizes of the sampled parts deviate from half of $|X|$, that is, we expect that $|X\cap \mathcal{U}_1|\ge (\frac{1}{2}+\Omega(\varepsilon))\cdot t$.
This gap should provide a substantial improvement over $\binom{n}{t}/\binom{k}{t}$.
We formalize this in the following technical lemma that lies at the heart of our algorithm.

\begin{lemma}[Biased sampling,]\label{lemma:partition-of-t}
Let integers $n\ge k\ge t\ge 1$ be given and let $\varepsilon\in(0,\frac{1}{2}]$. Let integers
$k_1\in\{0,1,\dots,k\},$ and $n_1\in\{0,1,\dots,n\}$
satisfy the one-sided conditions\\
$k_1 \ge \Bigl(\tfrac12+\tfrac\varepsilon2\Bigr)k,$ and
$n_1 \le k_1 + (n-k)\Bigl(\tfrac12-\tfrac\varepsilon2\Bigr).
$
Then \[
\frac{\binom{n_1}{t_1}}{\binom{k_1}{t_1}}\cdot\frac{\binom{n_0}{t_0}}{\binom{k_0}{t_0}}\le \frac{\binom{n}{t}}{\binom{k}{t}}\cdot \tfrac{20}{3}\cdot \sqrt{t}\cdot \sqrt{e}^{\!-t\varepsilon^2\Bigl(1-\frac{k}{n}\Bigr)^2},
\]
where $n_0=n-n_1$, $k_0=k-k_1$, $t_0=t-t_1$.
\end{lemma}
In the statement above, $n_i$ corresponds to $|\mathcal{U}_i|$ and $k_i$ corresponds to $|X\cap \mathcal{U}_i|$.
The conditions imposed on $n_1$ and $k_1$ comes from the discussion before that.
Essentially, \Cref{lemma:partition-of-t} gives us an exponential (in $t$) boost over simple sampling of $X$ over a single $(n,k,t)$-set-inclusion family.
Instead we will sample $X_i$ from a $(n_i,k_i,t_i)$-set-inclusion family for each $i\in\{0,1\}$ separately.

\begin{proof}[Proof of \Cref{lemma:partition-of-t}]
Let integers $n\ge k\ge t\ge 1$ be given and let $\varepsilon\in(0,1/4]$ (we use $\varepsilon$ instead of $\varepsilon/2$ for convenience). Let integers
\[
k_1\in\{0,1,\dots,k\},\qquad n_1\in\{0,1,\dots,n\}
\]
satisfy the one-sided conditions
\[
k_1 \ge \Bigl(\tfrac12+\varepsilon\Bigr)k,
\qquad
n_1 \le k_1 + (n-k)\Bigl(\tfrac12-\varepsilon\Bigr).
\]
For each integer $t_1\in\{0,1,\dots,t\}$ define
\[
R(t_1):=\frac{\binom{k}{k_1}\binom{n-t}{\,n_1-t_1\,}}{\binom{n}{n_1}\binom{k-t}{\,k_1-t_1\,}}=\frac{\binom{n_1}{t_1}}{\binom{k_1}{t_1}}\cdot\frac{\binom{n_0}{t-t_1}}{\binom{k_0}{t-t_1}}: \frac{\binom{n}{t}}{\binom{k}{t}},
\]
with the convention $R(t_1)=0$ if $\binom{k-t}{k_1-t_1}=0$.
Then we need to show that there exists $t_1\in\{0,1,\dots,t\}$ such that
\[
R(t_1)\le \frac{20}{3}\sqrt{t}\,\exp\!\Bigl(-2t\varepsilon^2\Bigl(1-\frac{k}{n}\Bigr)^2\Bigr).
\]

Fix any $t_1\in\{0,1,\dots,t\}$ such that $\binom{k-t}{k_1-t_1}>0$; otherwise $R(t_1)=0$ and the bound is immediate. Multiply and divide by $\binom{t}{t_1}$:
\[
R(t_1)=
\frac{\binom{t}{t_1}\binom{n-t}{n_1-t_1}/\binom{n}{n_1}}
     {\binom{t}{t_1}\binom{k-t}{k_1-t_1}/\binom{k}{k_1}}.
\]
Define random variables
\[
X\sim \mathrm{Hypergeo}(N=n,K=t,m=n_1),\qquad
Y\sim \mathrm{Hypergeo}(N=k,K=t,m=k_1),
\]
so that
\[
\Pr[X=x]=\frac{\binom{t}{x}\binom{n-t}{n_1-x}}{\binom{n}{n_1}},
\qquad
\Pr[Y=y]=\frac{\binom{t}{y}\binom{k-t}{k_1-y}}{\binom{k}{k_1}}.
\]
Then for all admissible $t_1$,
\[
R(t_1)=\frac{\Pr[X=t_1]}{\Pr[Y=t_1]}.
\]

Let $\mu_X:=\Ebb X$ and $\mu_Y:=\Ebb Y$. For the hypergeometric mean,
$\mu_X = \frac{t}{n}\,n_1,\ 
\mu_Y =\frac{t}{k}\,k_1.$
Using the one-sided constraint on $n_1$ we obtain
\[
\mu_X \le \frac{t}{n}\Bigl(k_1+(n-k)\bigl(\tfrac12-\varepsilon\bigr)\Bigr)
= t\Bigl(\frac{k_1}{n}+\Bigl(1-\frac{k}{n}\Bigr)\bigl(\tfrac12-\varepsilon\bigr)\Bigr).
\]
Therefore
\begin{align*}
\mu_Y-\mu_X \ge t\Bigl(\frac{k_1}{k}-\frac{k_1}{n}-\Bigl(1-\frac{k}{n}\Bigr)\bigl(\tfrac12-\varepsilon\bigr)\Bigr) = t\Bigl(1-\frac{k}{n}\Bigr)\Bigl(\frac{k_1}{k}-\bigl(\tfrac12-\varepsilon\bigr)\Bigr).
\end{align*}
Using $k_1/k\ge \tfrac12+\varepsilon$ yields
\[
\mu_Y-\mu_X \ge t\Bigl(1-\frac{k}{n}\Bigr)\Bigl(\bigl(\tfrac12+\varepsilon\bigr)-\bigl(\tfrac12-\varepsilon\bigr)\Bigr)
=2t\varepsilon\Bigl(1-\frac{k}{n}\Bigr).
\]
Denote
\[
\Delta:=2t\varepsilon\Bigl(1-\frac{k}{n}\Bigr),
\qquad\text{so that}\qquad
\mu_Y-\mu_X \ge \Delta.
\]

Now choose $t_1$ to be a mode of $Y$, i.e. an integer maximizing $\Pr[Y=y]$ over $y\in\{0,1,\dots,t\}$.
We lower bound $\Pr[Y=t_1]$ via variance and Chebyshev.

The variance of a hypergeometric variable satisfies
\[
\Var(Y)=k_1\cdot \frac{t}{k}\Bigl(1-\frac{t}{k}\Bigr)\cdot\frac{k-k_1}{k-1}\le t,
\]
since $1-\frac{t}{k}\le 1$, $\frac{k-k_1}{k-1}\le 1$, and $k_1\cdot\frac{t}{k}\le t$.
Let $\sigma:=\sqrt{\Var(Y)}\le\sqrt{t}$. By Chebyshev,
\[
\Pr[|Y-\mu_Y|\ge 2\sigma]\le \frac14,
\qquad\text{hence}\qquad
\Pr[|Y-\mu_Y|\le 2\sigma]\ge \frac34.
\]
The interval $\{y\in\mathbb{Z}:|y-\mu_Y|\le 2\sigma\}$ contains at most $4\sigma+1\le 4\sqrt{t}+1\le 5\sqrt{t}$ integers for $t\ge 1$.
Therefore, since $t_1$ maximizes the point mass,
\[
\Pr[Y=t_1]\ge \frac{3/4}{5\sqrt{t}}=\frac{3}{20}\cdot\frac1{\sqrt{t}}.
\]
Thus
\begin{equation}\label{eq:lowerPY_new}
\Pr[Y=t_1]\ge \frac{3}{20}\cdot\frac1{\sqrt{t}}.
\end{equation}

Next we upper bound $\Pr[X=t_1]$ by a right-tail bound.
First note $\Pr(X=t_1)\le \Pr[X\ge t_1]$.
We use the Hoeffding--Serfling inequality for sampling without replacement: if $X\sim\mathrm{Hypergeo}(N=n,K=t,m=n_1)$ and $\mu_X=\Ebb X$, then for any $u>0$,
\begin{equation}\label{eq:hoeffding_serfling_new}
\Pr[X\ge \mu_X+u]\le \ \!\Bigl(-\frac{2u^2}{t}\Bigr).
\end{equation}
It remains to lower bound $t_1-\mu_X$.
Since the hypergeometric pmf is log-concave and hence unimodal, every mode differs from the mean by at most $1$; in particular
$t_1\ge \mu_Y-1.$
Therefore
\[
t_1-\mu_X \ge (\mu_Y-1)-\mu_X \ge (\mu_Y-\mu_X)-1 \ge \Delta-1.
\]
If $\Delta\ge 2$, then $t_1-\mu_X\ge \Delta/2$, and applying \eqref{eq:hoeffding_serfling_new} with $u=t_1-\mu_X$ gives
\[
\Pr[X=t_1]\le \Pr[X\ge t_1]\le \Pr\Bigl[X\ge \mu_X+\frac{\Delta}{2}\Bigr]
\le \exp\!\Bigl(-\frac{2(\Delta/2)^2}{t}\Bigr)
=\exp\!\Bigl(-\frac{\Delta^2}{2t}\Bigr).
\]
Since $\Delta=2t\varepsilon(1-k/n)$,
\[
\frac{\Delta^2}{2t}=2t\varepsilon^2\Bigl(1-\frac{k}{n}\Bigr)^2,
\]
hence, for $\Delta\ge 2$,
\begin{equation}\label{eq:upperPX_new}
\Pr[X=t_1]\le \exp\!\Bigl(-2t\varepsilon^2\Bigl(1-\frac{k}{n}\Bigr)^2\Bigr).
\end{equation}
If $\Delta<2$, then the claimed bound in the statement holds after enlarging the absolute prefactor (since $R(t_1)\le 1$ for all admissible parameters, while the right-hand side in the statement is at least $(20/3)\sqrt{t}\,e^{-8/t}\ge 1$ for all $t\ge 1$); thus we may assume $\Delta\ge 2$ and use \eqref{eq:upperPX_new}.

Finally, combine $R(t_1)=\Pr[X=t_1]/\Pr[Y=t_1]$ with \eqref{eq:lowerPY_new} and \eqref{eq:upperPX_new}:
\[
R(t_1)\le
\frac{\exp\!\bigl(-2t\varepsilon^2(1-k/n)^2\bigr)}{(3/20)\,t^{-1/2}}
=\frac{20}{3}\sqrt{t}\,\exp\!\Bigl(-2t\varepsilon^2\Bigl(1-\frac{k}{n}\Bigr)^2\Bigr).
\]
If $k\le (1-\gamma)n$, then $(1-k/n)^2\ge \gamma^2$, giving the stated corollary bound.
\end{proof}

Before moving to the proof of \Cref{thm:noisy-mls} itself, we highlight the key difficulties that lie ahead.
On the one hand, \Cref{lemma:partition-of-t} does not give us any improvement if $k/n$ is too high, since the exponential speed up depends on $1-k/n$.
On the other hand, $k$ cannot be very small as well, because the probability that $X$ is distributed unevenly between $\mathcal{U}_0$ and $\mathcal{U}_1$ with a large enough deviation from $\frac{1}{2}$ is high only when $\varepsilon^2 \cdot k$ is large.
Therefore, we should handle the values of $k$ close to $n$ or close to $0$ in an alternative way not present in the original MLS.
This complicates the running time analysis and requires additional work to tweak the algorithm's alternative parts and reach the running time claimed in \Cref{thm:noisy-mls}. Armed up with this understanding, we start proving the main result of our work.

\begin{proof}[Proof of \Cref{thm:noisy-mls}.]

\input{pseudocode}
Let $\mathcal{B}$ be the parameterized algorithm for \textsc{$\Phi$-Extension} that works with a given instance $I$, a partial solution candidate $X\subset [n]$, and an integer $k'$, finds an $Y\subset [n]\setminus X$ of size \emph{exactly} $k'$ such that $X\cup Y$ is a solution to $I$, if it exists.

The algorithm $\mathcal{E}$, a prediction-guided exact algorithm that finds a solution subset of size exactly $k$ if it exists, is present in Algorithm~\ref{alg:noisy-mls} as a pseudocode.
It consists of two parts, the exhaustive search part for very small and very large values of $k$ (lines \ref{line:first-part-start}-\ref{line:first-part-end}), and the MLS part (lines \ref{line:mls-part-start}-\ref{line:mls-part-end}).
The preliminary part (lines \ref{line:prelim-part-start}-\ref{line:prelim-part-end}) evaluates the value $\gamma$ depending on $c$, which specifies the threshold between the alternative parts.
We move on to the analysis of the algorithm, that also contains some alternative discussion of how it works.

\paragraph{Correctness and success probability.}
The algorithm is deterministic, but may fail if the prediction quality is bad.
The exhaustive search part is ran if $k\le \gamma n$ or $n-k\le \gamma n$, it is obviously correct and does not use the prediction.
The second part of the algorithm relies on the prediction.
Let $S^*\subset \mathcal{U}$ be the solution of size $k$, which we have the prediction $(\mathcal{U}_0, \mathcal{U}_1)$ for in the form of partition of $[n]$.
For each $i\in\{0,1\}$, let $n_i=|\mathcal{U}_i|$ and let $k_i=|S^*\cap \mathcal{U}_i|$.
For $\mathcal{E}$ to work properly, we require that \Cref{lemma:partition-of-t} is applicable, that is, $k_1\ge \left(
\frac{1}{2}+\frac{\varepsilon}{2}\right)k$ and $n_1\le k_1+(n-k)(\frac{1}{2}-\frac{\varepsilon}{2})$.
The probability that is not true is then at most
\[\Pr\left[k_1\le\left(\tfrac{1}{2}+\tfrac{\varepsilon}{2}\right)\cdot k \lor |\mathcal{U}_1\setminus S^*|\ge |\mathcal{U}\setminus S^*|\cdot\left(\tfrac{1}{2}-\tfrac{\varepsilon}{2}\right)\right],\]
where $(\mathcal{U}_0, \mathcal{U}_1)$ is the only random variable. This is equivalent to
\[\Pr\left[|\mathcal{U}_1\cap S^*|\le\left(\tfrac{1}{2}+\tfrac{\varepsilon}{2}\right)\cdot |S^*| \lor |\mathcal{U}_0\setminus S^*|\le \left(\tfrac{1}{2}+\tfrac{\varepsilon}{2}\right)\cdot |\mathcal{U}\setminus S^*|\right],\]
which is upper bounded, via union bound, with
\[\Pr\left[|\mathcal{U}_1\cap S^*|\le\left(\tfrac{1}{2}+\tfrac{\varepsilon}{2}\right)\cdot |S^*|\right] +\Pr\left[|\mathcal{U}_0\setminus S^*|\le \left(\tfrac{1}{2}+\tfrac{\varepsilon}{2}\right)\cdot |\mathcal{U}\setminus S^*|\right].\]

Note that $|\mathcal{U}_1\cap S|$ and $|\mathcal{U}_0\setminus S|$ are sums of $k$ and $(n-k)$ random variables respectively, each variable is at least $\frac{1}{2}+\varepsilon$ in expectation.
In this part of the algorithm, we have $\min\{k, n-k\}\ge \gamma\cdot n$.
Then by \Cref{lemma:sum_near_expectation}, these probabilities are both at most $1/(\varepsilon^2\gamma n)$, under  pairwise independence, and at most $\exp(-\varepsilon^2\gamma n/2)$, under mutual independence of predictor's guess probabilities.
If both events do not occur, the algorithm finds the solution correctly for at least one choice of $k_1$: $\mathcal{E}$ has to iterate $k_1$ because it can't know it in advance.

\paragraph{Running time.}
We now give an upper bound to the running time of our algorithm.
To bound this, we use the following numerical lemma relating the binary entropy function to a simpler expression. We show that on the $(0, \frac{1}{2}]$ interval, the function $2^{H(x)}$ can be bounded by $1 + 2 \cdot x \log_2 \frac{1}{x}$. Moreover, this bound is tight, since $2^{H(\frac{1}{2})} = 2 .$  

\begin{lemma}
\label{lemma:1-2xlogx}
    For every $x \in (0, \frac{1}{2}]$, $2^{H(x)} \le 1 + 2 \cdot x \log_2(\frac{1}{x})$.
\end{lemma}
\begin{proof}
    On the interval $(0, \frac{1}{2})$, $- (1-x) \log_2(1-x) \le - x \log_2(x)$.
    Hence,
    $$H(x) = -x \log_2 x - (1 - x) \log_2(1 - x) < - 2x \log_2 x = 2x \log_2 \frac{1}{x} \; .$$
    Since $2^t \le 1 + t$ for $0 \le t \le 1$ and $0 \le H(x) \le 1$,
    $$2^{H(x)} \le 1 + H(x) \le 1 + 2x \log_2 \frac{1}{x} \; .$$
\end{proof}

Armed with this, we can now bound the running time of the first part. The exhaustive search part boils down to enumerating a subset of $[n]$ of size at most $\gamma n$, so the running time of this part is upper bounded by $2^{H(\gamma)\cdot n}\cdot\polyn$.
We claim that for both $c\le \frac{4}{3}$ and $c>\frac{4}{3}$, the two different cases of choosing $\gamma$, the running time fits within the bound claimed by \Cref{thm:noisy-mls}.

\begin{claim}\label{claim:mls-first-part-bound}
    $2^{H(\gamma)}<1+(1-\frac{1}{c})/2\le 2-\frac{1}{c}-(1-\frac{1}{c})\cdot \varepsilon$.
\end{claim}
\begin{proof}
    When $c>\frac{4}{3}$, we use $\gamma=0.05$.
    In this case, $1.219<2^{H(\gamma)}<1.22$.
    On the other hand, $1+(1-\frac{1}{c})/2> 1+0.75/2=1.375>2^{H(\gamma)}.$

    The other case is $c\le \frac{4}{3}.$
    Then $\gamma=(1-\frac{1}{c})/(20\ln(\frac{c}{c-1}))$.
    By \Cref{lemma:1-2xlogx}, $2^{H(\gamma)}$ is at most \[1+2\cdot \left(1-\tfrac{1}{c}\right)\cdot \log_2\left(\tfrac{c}{c-1} \cdot20 \cdot \ln\left(\tfrac{c}{c-1}\right)\right) /\left(20\cdot  \ln\left(\tfrac{c}{c-1}\right)\right).\]
This is equivalent, for $\rho=c/(c-1)\ge 4$, to
\[1+\left(1-\tfrac{1}{c}\right)\cdot\tfrac{1}{10}\cdot \left(\tfrac{\log_2\rho}{\ln\rho}+\tfrac{\log_2 20}{\ln\rho}+\tfrac{\log_2(\ln\rho)}{\ln\rho}\right).\]
For $\rho\ge 4$, 
$\tfrac{\log_2\rho}{\ln\rho}+\tfrac{\log_2 20}{\ln\rho}+\tfrac{\log_2(\ln\rho)}{\ln\rho}<5,$
and the claim follows.
\end{proof}

The analysis of the second part relies on two auxiliary lemmas. The first gives a useful upper bound on $c^{-x}$, and the second is an enhanced version of Lemma~2.3 from~\citep{MLS} (original version has $c_2=1$).
\begin{lemma}
\label{lemma:c**-x}
    For every $c > 1$ and $x \in [0, \frac{1}{\ln c}]$, $c^{-x} \le 1 - \frac{\ln c}{2} \cdot x$. 
\end{lemma}
\begin{proof}
    Let us define a function $f$ as
    $$f(x) = 1 - \frac{\ln c}{2} \cdot x - c^{-x} \; .$$
    We compute $f$ in $0$ and $\frac{1}{\ln c}$ and show that its second derivative is negative within the interval between them, meaning that for $x \in [0, \frac{1}{\ln c}]$, $f(x) \ge \min(f(0), f(\frac{1}{\ln c}))$.
    $$f'(x) = - \frac{\ln c}{2}  + c^{-x} \cdot \ln c \; .$$
    $$f''(x) =  - c^{-x} \cdot \ln^2 c \le 0 \; ,$$
    so $f$ is concave.
    Moreover,
    $$f(0) = 1 - 0 - 1 = 0 \; , \quad \text{and}$$
    $$f \big(\,\frac{1}{\ln c}\,\big) = 1 - \frac{\ln c}{2} \cdot \frac{1}{\ln c} - c^{- \frac{1}{\ln c}} = 1 - \frac{1}{2} - e^{-\frac{\ln c}{\ln c}} = 1 - 1/2 - 1/e > 0 \;, $$
    so for every $x \in [0, \frac{1}{\ln c}]$, $f(x) \ge 0$ implying $c^{-x} \le 1 - \frac{\ln c}{2} \cdot x$.
\end{proof}

\begin{lemma}\label{lemma:modified-local-search-bound}
    Let $c_1>c_2\ge 1$ be two constants and $n\ge k\ge t$ be three non-negative integers.
    Then $$\frac{\binom{n}{k}\cdot c_1^{k-t}}{\binom{n-t}{k-t}\cdot c_2^{k}}\le \left(1+\frac{1}{c_2}-\frac{1}{c_1}\right)^n\cdot \polyn.$$
\end{lemma}
\begin{proof}
    First, if $k\le n/c_1$, then pick $t=0$.
    The fraction is then at most $(\frac{c_1}{c_2})^{n/c_1}$, because $c_1>c_2\ge 1$, $\left(\frac{c_1}{c_2}\right)^{1/c_1}<1-\frac{1}{c_1}+\frac{1}{c_2}$.
    Now we know that $k>n/c_1$ and a lower bound for $\binom{n-t}{k-t}\cdot (\frac{1}{c_1})^{k-t}$, it is at least $\frac{1}{(1-\frac{1}{c_1})^{n-k}}$ up to a polynomial factor, for some choice of $t$.

    Then $$\frac{\binom{n}{k}\cdot c_1^{k-t}}{\binom{n-t}{k-t}\cdot c_2^{k}} \le \binom{n}{k}\cdot \left(\frac{1}{c_2}\right)^{k}\cdot\left(1-\frac{1}{c_1}\right)^{n-k}\le \left(1-\frac{1}{c_1}+\frac{1}{c_2}\right)^n.$$
\end{proof}

We are now ready to bound the running time of the MLS part. The second part of the algorithm uses set-inclusion families and the parameterized algorithm for \textsc{$\Phi$-Extension}.
For each choice of $k_1$, the running time of the algorithm's iteration is at most $|\mathcal{S}_0|\cdot |\mathcal{S}_1|\cdot c^{k-t}\cdot 2^{o(n)}$ (using the set inclusion family construction of \Cref{lemma:universal-family}).
We note that the choice of $t$ depends not only on $n, k$ and $c$, but also on $\gamma$ and $\varepsilon$ (line~\ref{line:choice-of-t}).
While $t$ corresponds to the size of $X$, we also choose $t_1$ corresponding to $|X\cap \mathcal{U}_1|$ according to our ``uneven parts'' intuition (line~\ref{line:choice-of-t1}).
We claim that there exist choices that allow to ``shave off'' some additional part (that depends on $c$, $\gamma$ and $\varepsilon$) from the original $(2-\frac{1}{c})$ exponent base of the MLS.

\begin{claim}\label{claim:main-runtime}
    Let $n_0,n_1,k\ge 0$ be integers such that $n_0+n_1=n$ and $k\le(1-\gamma)n$.
    For each $k_1\ge \left(\frac{1}{2}+\frac{\varepsilon}{2}\right)\cdot k$ with $k_1\le k$ and $n_1\le k_1+(n-k)\cdot (\frac{1}{2}-\frac{\varepsilon}{2})$, there exists a choice of $t_0,t_1\ge 0$ with $t_0+t_1\le k$ such that
    \[\tfrac{\binom{n_0}{t_0}}{\binom{k_0}{t_0}}\cdot\tfrac{\binom{n_1}{t_1}}{\binom{k_1}{t_1}} \cdot c^{k-(t_0+t_1)}\le \left(1+\left(1-\tfrac{1}{c}\right)\cdot (1-\tfrac{\gamma^2\varepsilon^2}{4})\right)^n\cdot \polyn.\]
\end{claim}
\begin{proof}
    By \Cref{lemma:partition-of-t}, when $t=t_0+t_1$ is fixed, we can always choose $t_1$ such that

    \[\tfrac{\binom{n_0}{t_0}}{\binom{k_0}{t_0}}\cdot\tfrac{\binom{n_1}{t_1}}{\binom{k_1}{t_1}}\le \tfrac{\binom{n}{t}}{\binom{k}{t}}\cdot e^{-\frac{1}{2} \gamma^2\varepsilon^2\cdot t}\cdot n^{\Oh(1)}.\]

    Note that $\binom{n}{t}/\binom{k}{t}=\binom{n}{k}/\binom{n-t}{k-t}$.
    Therefore, the upper bound above is equivalent (up to a polynomial factor) to
    \begin{equation}\label{eq:upper-bound-inner}
    \tfrac{\binom{n}{k}}{\binom{n-t}{k-t}\cdot e^{\frac{1}{2}\gamma^2\varepsilon^2\cdot t}} \cdot c^{k-t}=\tfrac{\binom{n}{k}\cdot c^{k-t}}{ \binom{n-t}{k-t}\cdot(A_{\gamma,\varepsilon})^t}=\tfrac{\binom{n}{k}\cdot c^{k-t}\cdot (A_{\gamma,\varepsilon})^{k-t}}{ \binom{n-t}{k-t}\cdot (A_{\gamma,\varepsilon})^k}=\tfrac{\binom{n}{k}\cdot (c\cdot A_{\gamma,\varepsilon})^{k-t}}{ \binom{n-t}{k-t}\cdot (A_{\gamma,\varepsilon})^k},
    \end{equation}
    where $A_{\gamma,\varepsilon}=e^{\frac{\gamma^2\varepsilon^2}{2}}>1.$
    
       Finally, to show that the choice of $t$ satisfying the claim statement exists, we apply \Cref{lemma:modified-local-search-bound} with $c_1:=cA_{\gamma,\varepsilon}$ and $c_2:=A_{\gamma,\varepsilon}$ to the right part of (\ref{eq:upper-bound-inner}).
       This yields the upper bound of form \[\left(1-\tfrac{1}{cA_{\gamma,\varepsilon}}+\tfrac{1}{A_{\gamma, \varepsilon}}\right)^n=\left(1+(1-\tfrac{1}{c})\cdot\tfrac{1}{A_{\gamma,\varepsilon}}\right)^n\le (1+(1-\tfrac{1}{c})\cdot (1-\tfrac{\gamma^2\varepsilon^2}{4})),\]
       where the last inequality comes from \Cref{lemma:c**-x}. The proof of the claim is complete.
\end{proof}

The final upper bound on the running time of the second part thus depends on the choice of $\gamma$.
In its turn, it depends on $c$.
If $c>\frac{4}{3}$, $\gamma=0.05$ and the bound in \Cref{claim:main-runtime} gives $(2-\frac{1}{c}+(1-\frac{1}{c})\cdot\Omega(\varepsilon^2))^n$, where $\Omega$ hides constant $\gamma^2$.
If $c\le \frac{4}{3}$, $\gamma^2$ is estimated as (by line~\ref{line:eval-gamma})
\[\gamma^2=\tfrac{1}{400}\cdot \left(1-\tfrac{1}{c}\right)^2/\ln^2\left(\tfrac{c}{c-1}\right)<\tfrac{1}{400}\cdot\left(1-\tfrac{1}{c}\right)^3,\]
since $\frac{1}{\ln^{2} \frac{1}{x}}>x$ holds for any $x\in (0,1)$.
Consequently, the running time of the second part is upper bounded by $(2-\frac{1}{c}+(1-\frac{1}{c})^4\cdot\Omega(\varepsilon^2))^n$, and we can safely hide the explicit $(1-\frac{1}{c})^4\cdot \Omega(\varepsilon^2)$ under the implicit $\Omega_c(\varepsilon^2)$ and finish the proof of the main result.
\end{proof}

\section{Unknown Prediction Quality}\label{sec:quality}

The previous section introduced algorithms that improve exhaustive and Monotone Local Search substantially under the presence of noisy predictors.
They are, however, susceptible to not knowing the estimate on the predictor's accuracy offset $\varepsilon$ (but they require its explicit value) and even the prediction quality, as it still may happen that a predictor gives prediction that deviates from the average too much.

In this section, we aim to deal with this problem.
Under mutual independence, we are able to achieve algorithms that work in average $(2-\Omega(\varepsilon^2))^n$ time (for exhaustive search) and in average $(2-\frac{1}{c}+\Omega_c(\varepsilon^2))^n$ time (for Monotone Local Search), if a prediction is given by a noisy predictor with \emph{unknown} accuracy $\frac{1}{2}+\varepsilon$.
The algorithms versions here, in contrast to previous sections, do not rely on the assumptions that prediction is good with high probability; instead the low probability of bad predictions is smoothed out by the average running time.
Another important assumption of our algorithms is that the predictor knows a viable solution that is actually a solution to the problem, so an algorithm can stop on its encounter.

The key idea is to deal with unknown value of $\varepsilon$ by making consecutive runs of the original algorithms from \Cref{sec:search}, with gradually decreasing choices of $\varepsilon$.
Each consecutive run will use exactly the same prediction as before, taken only once initially from the predictor.
For instance, take the prediction-guided algorithm of \Cref{subsec:exhaustive} and run it with (the guessed estimate) $\varepsilon_1=\frac{1}{4}$.
If it gives a solution, we can report it and stop.
Otherwise, repeat the run with $\varepsilon_2=\frac{1}{8}$, if it does not give the solution, repeat it again with $\varepsilon_3=\frac{1}{16}$ and so on, until the run finishes with a solution.
On $i^\text{th}$ run, we use $\varepsilon_k=2^{-(i+1)}$.
Note that when $i>\log_2n$, we can just use the (unbiased) brute-force that runs in $2^n\cdot\polyn$, because its only a polynomial-factor superior to $(2-\frac{1}{n})^n$.

To estimate the average running time of this algorithm, it is enough to estimate the probability that an iteration with $\varepsilon_k\le \varepsilon$ does not find a solution.
All prior iterations are negligible, since they run even faster than $(2-\Omega(\varepsilon^2))^n$ and they are at most $-\log_2\varepsilon$ many of them.
Note that the choice of $\varepsilon_k$ fails if the symmetric difference $S^*\triangle \widetilde{S}^*$ is larger than $(\frac{1}{2}+\frac{\varepsilon_i}{2})\cdot n$.
Similarly to \Cref{lemma:sum_near_expectation} (the mutual independence part), the probability that this happens is at most $\exp(-2(\varepsilon-\frac{\varepsilon_i}{2})^2n)$ (\Cref{lemma:sum_near_expectation} just used $\varepsilon_i=\varepsilon$).
Then the average runtime is at most \[\sum_{i = 1}^{\lceil \log_2 n \rceil} {\ \big(2 - \Omega(2^{-2(i+1)}) \big)^n}\cdot {\exp(2 \cdot \big(\varepsilon - {2^{-(i+2)}}\big)^2 n)}\le (2-\Omega(\varepsilon^2))^n.\]

We now make this argument precise. For clarity, we simplify the running times of the corresponding algorithms from \Cref{sec:search}, to just $(2-\varepsilon^2)^n$ and $(2-\frac{1}{c}-\varepsilon^2)^n$.
This does not break the correctness of our computations.
\begin{figure}
    \centering
    \includegraphics[width=0.9\linewidth]{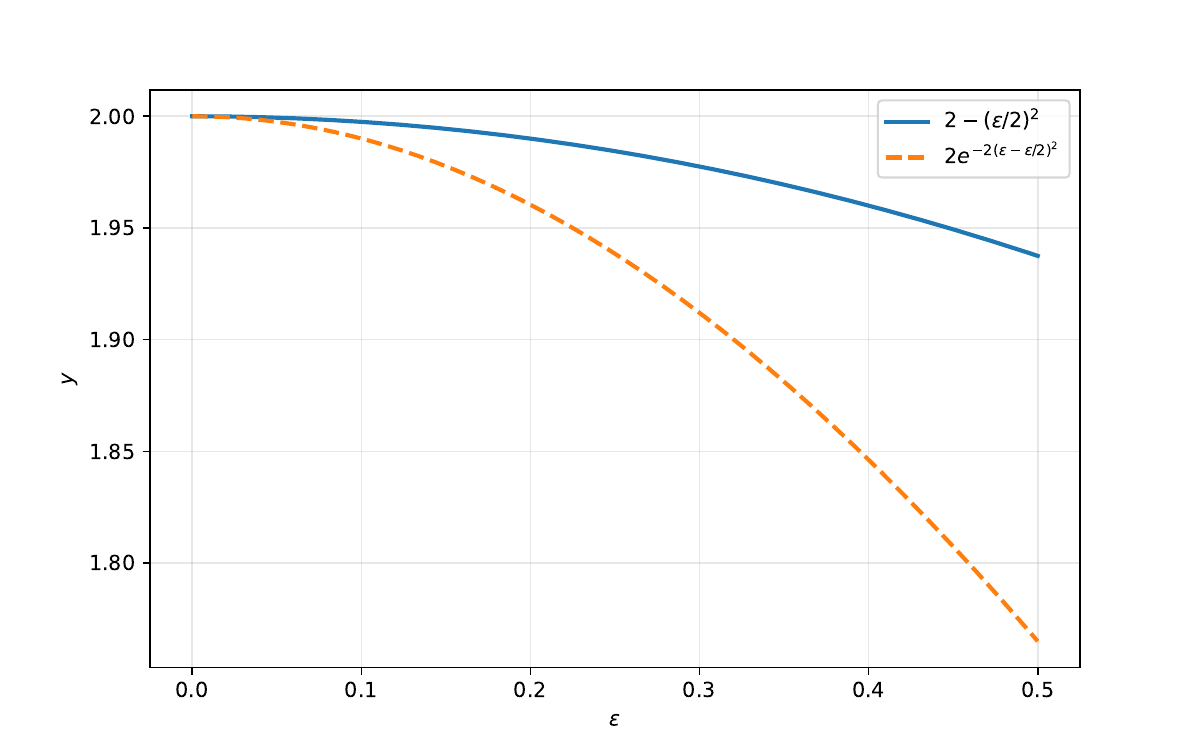}
    \caption{$2 - \varepsilon^2/4$ and $2 e^{-\varepsilon^2/2}$ on $[0, \frac{1}{2}]$}
    \label{fig:graph_max_func}
\end{figure}

\begin{lemma}
    The expected runtime of the prediction-guided exhaustive search algorithm is $(2 - \frac{\varepsilon^2}{4})^n\cdot\polyn$.
\end{lemma}
\begin{proof}
    We can estimate the expected value of the running time as    
    $$\sum\limits_{i = 1}^{\lceil \log_2 n \rceil} \frac{2 \cdot \big(2 - \big(\frac{1}{2^i}\big)^2 \big)^n}{e^{2 \cdot \big(\varepsilon - \frac{1}{2^i}\big)^2 n}} \; .$$
    Fix a threshold $t \in [0, 1]$, and consider the terms with indices before and after $-\log_2 t$ separately:
    \[
    \frac{2 \cdot \big(2 - \big(\frac{1}{2^i}\big)^2 \big)^n}{e^{2 \cdot \big(\varepsilon - \frac{1}{2^i}\big)^2 n}} \le 
    \begin{cases}
        (2 - t^2)^{n+1} \quad \text{for $\frac{1}{2^i} \ge t$ ,} \\
        2 \cdot 2^n / e^{^{2 (\varepsilon - t )^2n}} \quad \text{for $\frac{1}{2^i} < t$ .} 
    \end{cases}
    \]
    Set $t = \varepsilon/2$.
    As we can see in \Cref{fig:graph_max_func}, $2 - t^2 \ge 2 e^{-2(\varepsilon - t)^2}$ for $0 \le \varepsilon \le \frac{1}{2}$, which implies that 
    \[
    2 \cdot (2 - t^2)^n \ge 2 \cdot 2^n / e^{2 (\varepsilon - t)^2n}
    \]
    so the expected runtime of the algorithm can be bounded by $\mathcal{O^*}\big((2 - \frac{\varepsilon^2}{4})^n \big)$.  
\end{proof}

Note that pairwise independence does not allow us to reach this running time because the probability of bad-quality prediction drops slowly, only inversely proportional to $n$.
Thus the probability that the algorithm reaches $i>\log_2 n$ is not as low as we need to beat $2^n$.

The modified version of the prediction-guided MLS (\Cref{subsec:mls}) is attained in exactly the same way.
\begin{lemma}
    The expected runtime of the prediction-guided MLS algorithm is $\big((2 - \frac{1}{c} - \frac{1}{8} \cdot (2 - \frac{1}{c}) \cdot \varepsilon^2)^n\cdot\polyn$.
\end{lemma}
\begin{proof}
    We can estimate the expected value of the running time as    
    $$\sum\limits_{i = 1}^{\lceil \log_2 n \rceil} \frac{2 \cdot \big(2 - \frac{1}{c} - \big(\frac{1}{2^i}\big)^2 \big)^n}{e^{2 \cdot \big(\varepsilon - \frac{1}{2^i}\big)^2 n}} \; .$$
    Fix a threshold $t \in [0, 1]$, and consider the terms with indices before and after $-\log_2 t$ separately:
    \[
    \frac{2 \cdot (2 - \frac{1}{c} - \big(\frac{1}{2^i}\big)^2)^n}{e^{2 \cdot \big(\varepsilon - \frac{1}{2^i}\big)^2 n}} \le 
    \begin{cases}
        (2 - \frac{1}{c} - t^2)^{n+1} \quad \text{for $\frac{1}{2^i} \ge t$ ,} \\
        2 \cdot (2 - \frac{1}{c})^n / e^{^{2 (\varepsilon - t )^2n}} \quad \text{for $\frac{1}{2^i} < t$ .} 
    \end{cases}
    \]
    As we can see in \Cref{fig:graph_max_func}, $2 - \varepsilon^2/4 \ge 2e^{-\varepsilon^2/2}$ for $0 \le \varepsilon \le \frac{1}{2}$.
    We multiply this inequality by $(2 - \frac{1}{c})/2$, and get
    \[
    2 - \frac{1}{c} - \big(\frac{1}{4} - \frac{1}{8c}\big) \cdot \varepsilon^2 \ge (2 - \frac{1}{c}) \cdot e^{-2(\varepsilon- \varepsilon/2)^2} \;.
    \]
    Thus, by setting $t = \varepsilon/2$, we can bound the runtime of the algorithm by
    \[
    \mathcal{O^*}\big((2 - \frac{1}{c} - \frac{1}{8} \cdot \big(2 - \frac{1}{c}\big) \cdot \varepsilon^2)^n\big) \; .
    \]
\end{proof}

\section{Conclusion}\label{sec:concl}
We have initiated the study of learning-augmented exact exponential
algorithms, showing that even a noisy predictor only marginally better than random guessing suffices to provably improve the state-of-the-art exact algorithms for a broad class of NP-hard subset selection problems.
Our approach integrates predictions into exhaustive search and monotone local search, yielding improved running times for an entire family of problems simultaneously. The assumptions on the predictor are minimal: pairwise independence of prediction errors suffices for our main results, and the speedup scales smoothly with prediction quality. Beyond the concrete results, our work reveals an unexpected information leverage effect: a linear amount of predictive signal suffices to tame an exponential search.

% We have initiated the study of learning-augmented exact exponential algorithms, demonstrating that even weak predictions can break the computational barriers that govern the best known exact algorithms for a broad class of \classNP-hard subset selection problems. 
% Our general framework integrates a noisy predictor into the monotone local search paradigm, yielding improved running times for an entire family of problems simultaneously---without any problem-specific modifications.

% Two features of our approach are particularly noteworthy. First, the improvement is robust: it holds under the minimal assumption of \emph{pairwise independence} of prediction, which is strictly weaker than the mutual independence assumed in all prior work on predictions for \classNP-hard problems, or does not require knowledge of
% the predictor's accuracy~$\varepsilon$. 
% Second, the improvement is universal across prediction quality: any predictor with a non-trivial advantage $\varepsilon > 0$ over random guessing suffices, and the speedup scales smoothly with~$\varepsilon$.

% Our results also carry a conceptual message. The same $n$ bits of predictive information that suffice to improve approximation ratios for polynomial-time algorithms are shown here to tame an exponential search---yielding provably faster \emph{exact} solutions for \classNP-hard problems. 
% This reveals an unexpected informational leverage: a linear amount of predictive signal controls an exponential computation.

\bibliography{LAA}
\bibliographystyle{apalike}

\end{document}

%% file: defs.tex
\usepackage{graphicx,color}
\usepackage{boxedminipage}
\usepackage{tcolorbox} 

\usepackage{lmodern}
\usepackage{amsmath,amssymb,amsthm}
\usepackage{amsmath,amsfonts}
\usepackage{mathtools}

\usepackage{mdframed}
\usepackage{xifthen}
\usepackage{framed}
\usepackage{tabularx}
\usepackage{todonotes}
\usepackage{amssymb}
\usepackage{cleveref}
\usepackage[ruled,vlined,linesnumbered]{algorithm2e}
\usetikzlibrary{calc}

\newtheorem{example}{Example}
\newcommand{\Oh}{\mathcal{O}}

\newcommand{\Ebb}{\mathbb{E}}
\newcommand{\Var}{\mathrm{Var}}

\DeclareMathOperator{\operatorClassNP}{{\sf NP}}
\newcommand{\classNP}{\ensuremath{\operatorClassNP}}

\DeclareMathOperator{\operatorClassFPT}{{\sf FPT}\xspace}
\newcommand{\classFPT}{\ensuremath{\operatorClassFPT}\xspace}

\newlength{\RoundedBoxWidth}
\newsavebox{\GrayRoundedBox}
\newenvironment{GrayBox}[1]%
{\setlength{\RoundedBoxWidth}{.9\textwidth}
	\def\boxheading{#1}
	\begin{lrbox}{\GrayRoundedBox}
		\begin{minipage}{\RoundedBoxWidth}}%
		{   \end{minipage}
	\end{lrbox}
	\begin{center}
		\begin{tikzpicture}%
			\node(Text)[draw=black!20,fill=white,rounded corners,%
			inner sep=2ex,text width=\RoundedBoxWidth]%
			{\usebox{\GrayRoundedBox}};
			\coordinate(x) at (current bounding box.north west);
			\node [draw=white,rectangle,inner sep=3pt,anchor=north west,fill=white] 
			at ($(x)+(6pt,.75em)$) {\boxheading};
		\end{tikzpicture}
\end{center}}     

\newenvironment{defproblemx}[2][]{\noindent\ignorespaces%
	\FrameSep=6pt%
	\parindent=0pt%
	\vspace*{-1.5em}
	\ifthenelse{\isempty{#1}}{%
		\begin{GrayBox}{#2}%                
		}{%
			\begin{GrayBox}{#2 parameterized by~{#1}}%  
			}
			\begin{tabular*}{\textwidth}{@{\hspace{.1em}} >{\itshape} p{1cm} p{0.8\textwidth} @{}}%        
			}{
			\end{tabular*}%
		\end{GrayBox}%
		\ignorespacesafterend
	}  
	
	\newcommand{\defparproblema}[4]{% FJR Version
		\begin{defproblemx}[#3]{#1}
			Input:  & #2 \\
			Task: & #4
		\end{defproblemx}
	}%

\newcommand{\polyn}{n^{\Oh(1)}}

%% file: pseudocode.tex
\begin{algorithm}[h]
	%\DontPrintSemicolon
	\SetKwIF{IfOr}{OrIf}{ElseOr}{if}{or}{or}{}{} 
    \eIf{$c\le \frac{4}{3}$  \label{line:prelim-part-start}}
    {
    $\gamma\gets (1-\frac{1}{c})/(20\cdot\ln (\frac{c}{c-1}))$\label{line:eval-gamma}\;
    }{
    $\gamma\gets 0.05$\label{line:prelim-part-end}\; 
    }
    \If{$k<\gamma \cdot n$ {\bf or} $k>(1-\gamma)\cdot n$ \label{line:first-part-start}}{
        \ForEach{$X\subset \mathcal{U}$ of size $k$}{
            \lIf{$X$ is a solution}{\Return{$X$}}
        }
        \Return{\textsc{None}}\label{line:first-part-end}\;
    }
    $A_{\gamma,\varepsilon}\gets e^{\frac{1}{2}\gamma^2\varepsilon^2}$\;\label{line:mls-part-start}
    \ForEach{$k_1\in \left[\lceil(\frac{1}{2}+\varepsilon)\cdot k\rceil, k\right]$}{
        \lIf{$n_1>k_1+(n-k)\cdot (\frac{1}{2}-\frac{\varepsilon}{2})$}{\textbf{continue}}
        $k_0\gets k-k_1$\;
        $t\gets \operatorname{argmax}_{t=0}^k \left\{\binom{n-t}{k-t}\cdot \left(c\cdot A_{\gamma,\varepsilon}\right)^t\right\}$\label{line:choice-of-t}\;
        $t_1\gets \operatorname{argmin}_{t_1=0}^{t}\left\{\frac{\binom{n_0}{t-t_1}}{\binom{k_0}{t-t_1}}\cdot\frac{\binom{n_1}{t_1}}{\binom{k_1}{t_1}}\right\}$\label{line:choice-of-t1}\;
        $t_0\gets t-t_1$\;
        \ForEach{$i\in\{0,1\}$}{$\mathcal{S}_i\gets$ $(n_i, k_i, t_i)$-set-inclusion family over $\mathcal{U}_i$\;}
        \ForEach{$X_0\in \mathcal{S}_0$}{
            \ForEach{$X_1\in \mathcal{S}_1$}{
                $X\gets X_0\cup X_1$\;
                Run $\mathcal{B}$ on $(I,X,k-t)$\;
                \lIf{it outputs $Y$ such that $X\cup Y$ is a solution}{\Return{$X\cup Y$}}
            }
        }
    }
    \Return{\textsc{None}}\tcp*{no solution has been found}\label{line:mls-part-end}
	
	\caption{An exact algorithm for \textsc{$\Phi$-Subset} in presence of an FPT-algorithm $\mathcal{B}$ for \textsc{$\Phi$-Extension} and a prediction $(\mathcal{U}_0, \mathcal{U}_1)$ given by a predictor with accuracy $\frac{1}{2}+\varepsilon$.}\label{alg:noisy-mls}
\end{algorithm}